\documentclass[letterpaper,english,numberwithinsect]{amsart}
\usepackage{microtype}
\usepackage{amsmath,amsfonts,amssymb,xcolor}
\usepackage{epsfig}

\usepackage[T1]{fontenc}
\usepackage{verbatim}
\usepackage{gensymb}
\usepackage{mathrsfs}
\usepackage{amsbsy}
\usepackage{comment}
\usepackage{textcomp}
\usepackage{framed}
\usepackage{epsfig}
\usepackage{xspace}
\usepackage{subcaption}
\usepackage[foot]{amsaddr}

\usepackage[section]{placeins} 

\newtheorem{theorem}{Theorem}

\newtheorem{lemma}{Lemma}[section]
\newtheorem{corollary}{Corollary}[section]
\newtheorem{definition}{Definition}[section]

\newcommand{\Z}{\ensuremath{\mathbb Z}}
\newcommand{\cC}{\mathcal{C}}

\newcommand{\cM}{\mathcal{M}}

\newcommand{\cS}{\mathcal{S}}

\newcommand{\union}{\cup}

\newif\ifcomment
\commenttrue  

\title[Slow Convergence of Spin Glass Models]{Slow Convergence of Ising and Spin Glass Models with Well-Separated Frustrated Vertices}
\author{David Gillman}
\address{Department of Computer Science, New College of Florida, [Sarasota, FL, 34243], USA}
\email{dgillman@ncf.edu}
\author{Dana Randall}
\address{School of Computer Science, Georgia Institute of Technology, [Atlanta, GA 30332], USA}\email{randall@cc.gatech.edu}
\subjclass{Theory of computation $\rightarrow$ Randomness, geometry and discrete structures $\rightarrow$ Random walks and Markov chains} 

\begin{document}
\date{}
\thispagestyle{empty}

\begin{abstract}
Many physical models undergo phase transitions as some parameter of the system is varied. This phenomenon has bearing on the convergence times for local Markov chains walking among the configurations of the physical system.  One of the most basic examples of this phenomenon is the ferromagnetic Ising model on an $n\times n$ square lattice region $\Lambda$ with mixed boundary conditions. For this spin system, if we fix the spins on the top and bottom sides of the square to be $+$ and the left and right sides to be~$-$, a standard Peierls argument based on energy shows that below some critical temperature~$t_c$, any local Markov chain $\mathcal{M}$ requires time exponential in $n$ to mix. 

Spin glasses are magnetic alloys that generalize the Ising model by specifying the strength of nearest neighbor interactions on the lattice, including whether they are ferromagnetic or antiferromagnetic.  Whenever a face of the lattice is bounded by an odd number of edges with ferromagnetic interactions, the face is considered {\it frustrated} because the local competing objectives cannot be simultaneously satisfied.  We consider spin glasses with 
exactly four well-separated frustrated faces that 
are symmetric around the center of the lattice region under $90$ degree rotations.
We show that local Markov chains require exponential time for all spin glasses in this class.  This argument extends to the ferromagnetic Ising model with mixed boundary conditions described above, which behaves like spin glasses with frustrated faces on the boundary. The standard Peierls argument breaks down when the frustrated faces are on the interior of $\Lambda$ 
and yields weaker results when they are on the boundary of $\Lambda$ but not near the corners.  
We show that there is a universal temperature $T$ below which $\mathcal{M}$ will be slow for all spin glasses with four well-separated frustrated faces. Our argument shows that there is an exponentially small cut indicated by the {\it free energy}, carefully exploiting both  entropy and  energy to establish a small bottleneck in the state space to establish slow mixing.  

\end{abstract} 

\maketitle

\newpage
\setcounter{page}{1}

\section{Introduction}
The celebrated Ising model on the Cartesian lattice is a
fundamental  model for ferromagnetism and one of the simplest models demonstrating an order-disorder phase transition. Each configuration
$\sigma$ in the state space $\Omega= \{-1,+1\}^{n^2}$ consists of an
assignment of a $+$ or~$-$ spin to each of the vertices, and the {\it
 Gibbs (or Boltzmann) distribution} assigns weight 
$$ \pi(\sigma) = e^{-\beta H(\sigma)}/Z(\beta),$$
where 
$$H(\sigma) = - \sum_{(i,j)\in E} \sigma_i \sigma_j $$ is the {\it Hamiltonian} (or energy) of the system, $\beta = 1/T$ is inverse temperature, and 
 $Z(\beta) = \sum_{\sigma \in \Omega} e^{-\beta H(\sigma)}$ is the normalizing constant known as the {\it partition function}.  
In Sections~\ref{sec:cross} and~\ref{sec:slow} it will be convenient to write the probability of a configuration in terms of $\lambda = e^{2\beta} = e^{2/T}$, where $\lambda$ can be seen as the weight assigned to edges whose endpoints are assigned like spins.

Physicists characterize when there is a phase transition in a physical model 
by asking whether there is a unique limiting conditional distribution on finite subregions as the lattice size grows. The Gibbs distribution is defined as any limiting measure, but this limit might not be unique.  For example, for the Ising model on $\Z^2$ at sufficiently low temperatures, the probability of an interior vertex being assigned $+$ will be much higher if the boundary vertices were hard-wired to be $+$ than if they were hard-wired to be $-$, and this difference  persists in the limit.
The infinite volume Ising model was solved exactly by Onsager in 1944 \cite{onsager}, showing that there
is a critical value $\beta_c = \ln(1+\sqrt{2})/2$ such that for $\beta < \beta_c$ (i.e., high temperature),
the limiting distribution is unique, and for $\beta > \beta_c$ (i.e., low temperature),
spins on the boundary
of the region persist and there are multiple
limiting distributions. 
The all-plus and the all-minus boundary conditions are known to be extremal measures \cite{aiz, hig}.

A related effect
has been observed in the context of mixing times of local Markov chains for the Ising model on finite lattice regions with free boundaries (i.e., boundary vertices can take on either spin).
The {\it mixing time} $\tau({\mathcal M})$ of a chain ${\mathcal{M}}$, i.e., the number of
steps required so that the  distribution over configurations
is close to its stationary distribution, also undergoes 
a  phase change. 
When $\beta$ is  small, local
dynamics are known to be efficient \cite{mo1,  mo2, lmst}, while when $\beta$ is large, local
chains require exponential time to converge to equilibrium~[\ref{thomas}]. 
At low enough temperature, the Gibbs
distribution strongly favors configurations that have predominantly one
spin, and it will take exponential time to move from a mostly $+$ state to a mostly
$-$ one using moves that only change $o(n^2)$ sites at a time~\cite{mart}.


Mixing times of Markov chains are known to be sensitive to boundary conditions.
For example, local  chains on Ising configurations are conjectured to converge in polynomial time at all temperatures for the ``all $+$'' boundary condition where all vertices on the boundary are hard-wired to have $+$ spins.  
While still open,  
Martinelli \cite{martinelli94} showed mixing is indeed sub-exponential at all temperatures with all + boundary conditions and subsequently Lubetsky et al. \cite{lmst} showed that the chain converges in quasi-polynomial time.  
However, a standard Peierls argument can be used to show that when there are mixed boundary conditions with 4 connected components of like spins on the boundary, alternating ``$+, -, +, -$'', then the chain again will be slow at low temperatures.  In particular, for {\it mixed boundary conditions} where we fix the boundary to be + on the vertical sides and $-$ on the horizontal sides, then the chain provably requires time exponential in $n$ at sufficiently low temperature.  
For ``p-shifted mixed boundary conditions'' where we rotate the mixed  boundary conditions clockwise by $p$ units,
We explain this in Section~\ref{sec:cross}.  More powerful machinery such as the approach of Dobrushin, Koteck\'y and Schlosman 
\cite{dks} for the Ising model establish bounds on the temperature below which convergence is slow, but they do not easily extend to other cases we consider.

Similar questions can be examined in the context of {\it spin glasses},  or  magnetic alloys that are a natural generalization of the ferromagnetic and antiferromagnetic Ising models.  We are given a graph $G=(V,E)$ and a set of couplings $J_{ij} 
\in \{-1, +1\}$ for each edge $(i,j) \in E$.  The state space is $\Omega = \{-1,+1\}^V,$ where a configuration assigns a spin to each vertex in $V$.  For a spin glass configuration $\sigma \in \Omega,$ the Hamiltonian is defined as
$$H(\sigma) = - \sum_{(i,j)\in E} J_{ij} \sigma(i)\sigma(j)$$ 
and the Gibbs distribution is defined as for the Ising model as $\pi(\sigma) = e^{-\beta H(\sigma)}/Z(\beta)$.

When all the $J_{ij}=+1,$ this model is precisely the ferromagnetic Ising model on $G$; when all the $J_{ij}=-1,$ it is  antiferromagnetic.  In general, the behavior of a spin glass is much richer than simple models of magnetism because of the presence of {\it frustration}, or  competition between local interactions. 
In the case of $G = \Lambda$, a square region in the lattice, a face of $\Lambda$ is {\it frustrated} when $J_{ij}=-1$ for an odd number of edges around the face. No setting of the sites on the corners of such a face will satisfy all four edges, i.e., make each $J_{ij} \sigma(i)\sigma(j) = 1$. Even finding the ground states (or most likely configurations) reduces to solving an optimization problem that can be NP-hard (see, e.g., \cite{baharona}).  It will be convenient to refer to the dual lattice $\overline{\Lambda}=(\overline{V}, \overline{E})$ and refer to a frustrated face $f$ of $\Lambda$ by the frustrated vertex $v = \overline{f}$ in $\overline{V}$.  

Here, we study spin glasses with exactly four well-separated frustrated faces in order to understand the long-range interactions and their effects on mixing times.  We fix the nearest-neighbor interactions around the boundary of $\Lambda$ to be ferromagnetic, and we assume fixed $+$ sites on the boundary. 
Similar models with well-separated defects have been explored to understand long-range correlation; for example, in seminal work, Ciucu \cite{ciucu} studied the monomer-dimer model with a constant number of monomers and established a connection with electrical networks, settling a nearly century old conjecture about long-range effects due to isolated monomers.  
Similar questions arise naturally in the context of spin glasses.

We show that there is a universal temperature $T$ below which the Markov chain $\mathcal{M}$ will be slow for any spin glass with exactly four frustrated vertices that are well-separated and symmetric around the origin under $90$ degree rotations.  We identify a bottleneck in the state space by looking at the how the  {\it free energy} (i.e., $\ln Z /n^2$) changes as a parameter of the system is varied.  The same argument easily extends to the Ising model with $p$-shifted mixed boundary conditions, which behaves like spin glasses with four symmetric frustrated faces near the boundary (and  indeed can be viewed as a special case of the spin glasses we consider if we also fix $+$ spins adjacent to the boundary).

\begin{theorem}\label{thm:main}
Let $\Lambda$ be a square lattice region
with fixed $+$ sites on the boundary of $\Lambda$ and a fixed ferromagnetic interaction $J_{ij} = 1$ on each boundary edge $(i, j)$.
Suppose $\Lambda$ has  exactly four frustrated faces,  $f_1, \dots, f_4$, that are symmetric around the center of the lattice region under 90 degree rotations and are well-separated (i.e., the shortest lattice path from $f_i$ to $f_{i+1}$ has length $2n$, $i = 1, 2, 3$).  
Then there is a universal temperature 
$T$
such that  the Glauber dynamics ${\mathcal M}$ for the spin glass model on $\Lambda$ with $f_1,...,f_4$ the faces with frustration has mixing time  $\tau({\mathcal{M}}) \geq e^{cn}$, for some constant $c>0,$ whenever $t < T$.
\end{theorem}
\noindent 
The theorem remains true under the additional assumption of fixed $+$ sites adjacent to the boundary.  As a corollary this  gives a universal bound on the temperature for the Ising model with $p$-shifted mixed boundary conditions.

The proof of Theorem~\ref{thm:main} requires several innovations.
The standard argument to show slow mixing is based on the {\it conductance} of the Markov chain.  The key is showing that the state space $\Omega$ can be partitioned into two sets, $S$ and its complement $S^C$, such that getting from $S$ to some subset $S^C$ requires passing through a small cutset  $\cC \subset S^C$, and the stationary weights $\pi(S)$ and $\pi(S^C)$ are both exponentially larger than $\pi(\cC)$.  This establishes that the chain has low conductance, which implies it takes exponential time to converge to equilibrium \cite{js}.   
The main ingredient is typically a {\it Peierls argument} \cite{peierls}, which introduces a map $\Psi$ from $\cC$ to $S \cup S^C$.  Typically $\Psi$ is chosen so that for all $\sigma \in \cC$, we have $\pi(\Psi(\sigma)) \geq \pi(\sigma) e^{cn}$, mapping elements of $\cC$ to configurations with exponentially larger weight.  If we can show that $\Psi$ is nearly injective  (i.e., the cardinality of the inverse image of each configuration is bounded by a polynomial), then we can conclude that $\pi(\cC)$ is exponentially small.

In our setting, there is not always  a natural candidate map that increases the probability of a configuration exponentially. In fact, the standard map gives no guaranteed increase to the stationary probability when each side of the boundary has close to an equal number of $+$ and $-$ spins (when $p=0.5$ and the boundary changes spin at the center of the four sides of the boundary).  In this case, we exploit the low {\it entropy} of $\cC$ by defining an injective map from $\cC \times 2^{cn} \rightarrow \Omega$, for some $c>0$.  The map never decreases the weight of a configuration, so we again can conclude that $\pi(\cC)$ is exponentially small.  As we vary $p$, the {\it free energy} of $\cC$ remains small  compared to the two sides of the cut due to a derease in  energy (when $p$ is close to $0$) or due to entropy  (when $p$ is close to $0.5$); all other cases rely on both.

An important technical contribution in our proofs is in the construction of a new injective map.  The {\it contour representation} of a spin glass configuration consists of edges in the dual lattice that cross  edges $e=(i,j)$ where $J_{ij} \sigma(i) \sigma(j) = -1$; in this representation the frustrated vertices in the dual lattice have odd degree and all other vertices have even degree. Because of this property the contour representation can be decomposed into a even cycles (closed contours) and two long paths whose endpoints are the four frustrated vertices.  In the standard case of the Ising model  with mixed side boundary conditions, we can define an injective map that shifts the paths connecting the four frustrated vertices to paths with much shorter length, and therefore much larger probability.  The new paths can be added along the boundary by shifting closed contours.  In our case we cannot do this since we cannot always construct maps to configurations with larger probability. Therefore we define a map to a {\it set} of configurations of {\it at least equal} probability.  To complete the proof we require a careful map that allows us to reconstruct the original path, the new path, and the closed contours that are intersected when the new path is added.  Verifying that the map is injective now requires a very sensitive combinatorial encoding and decoding that is likely of independent interest.

\section{Preliminaries}\label{PrelimSection}

We review some standard background on Markov chains, convergence times, and the Ising model that are required for our results.

\subsection{Markov chains and mixing times}

Let ${\cM}$ be an ergodic, 
reversible Markov chain with arbitrary finite state space $\cS$, transition
probability matrix $P$, and stationary distribution $\pi$. Let
$P^t(x,y)$ be the $t$-step transition probability from $x$ to $y$, and
let $||\cdot, \cdot||$ denote total variation distance.

\begin{definition}
For $\varepsilon > 0$, the {\em mixing time} is defined as
$$\tau(\epsilon)=\min\{t : \max_{x \in \cS}  \sum_{y\in \cS} ||P^{t'}(x,y),\pi(y)||\leq \epsilon, \ \rm{for \ all} \  t'\geq t\}.$$
\end{definition}

\noindent A Markov chain is \emph{rapidly} (or {\it polynomially) mixing}  if the mixing
time is bounded above by a polynomial in $\log{\cS}$, the length of a description of a state in $\cS$.
A chain is \emph{slowly mixing} if the mixing time is bounded below by an exponential
function. 
The {\it conductance}, introduced by Jerrum and Sinclair \cite{js}, is useful to bound the
 mixing time \cite{js}. 

\begin{definition}
For a Markov chain with stationary distribution $\pi$,  the
\emph{conductance} $\Phi$ is 
$$\Phi =\min_{S: 0 < \pi(S)\leq 1/2} \frac{\sum_{x\in S, y\not\in S}
 \pi(x) P(x,y)}{\pi(S)}.$$ 
\end{definition}


\begin{theorem}
  \label{CondThm}
  {\rm (Jerrum and Sinclair \cite{js})}
  The mixing time of a Markov chain with conductance
  $\Phi$ satisfies:
  \[\tau(\epsilon) \ \geq \
  \left( \frac{1-2\Phi}{2\Phi} \right)\ln \epsilon^{-1}.\]
\end{theorem}

\noindent To establish slow mixing, our strategy will be to define a set $S$ along with sets $T \subset S^C$ and $\cC \subset S^C \setminus T$ in the state space, such that $\pi(S) = \pi(T)$ and $\pi(\cC)/\pi(S) < e^{-cn}$ and such that getting from $S$ to $S^C$ in the Markov chain requires going through $\cC$. 

In this paper, we will focus on the simplest local Markov chain ${\mathcal M}$ for the Ising and spin glass models, known as {\it Glauber dynamics}, which connects pairs of configurations whose spins differ on at most one vertex. In a given step, the chain picks any vertex $v \in \Lambda$ at random and changes the spin with the appropriate transition probabilities so that the chain converges to the Gibbs distribution $\pi$. For our  models, the transition probabilities of ${\mathcal M}$ are defined as
$$P(\sigma, \tau) \ = \ \frac{1}{2 n^2} \ \min \left( 1, \frac{\pi(\tau)}{\pi(\sigma)}\right),$$
if $|\{i: \sigma_i \neq \tau_i\}| = 1$, and with all remaining probability stay at the current configuration.

\vspace{-.1in}
\subsection{The Contour representation of the Ising and spin glass models}\label{sec:contour}
It will be convenient to view Ising and spin glass configurations in terms of  {\it contours}.  
Recall the setting of Theorem~\ref{thm:main} in which $\Lambda = (V, E)$ is a square lattice region and $\Omega = \{-1,+1\}^V$.
For every configuration $\sigma \in \Omega$, there is a contour representation $\Gamma(\sigma)$ in $\overline{\Lambda}$, the planar dual to $\Lambda$.  
We define $\overline{\Lambda}=(\overline{V}, \overline{E})$ by letting $\overline{V}$ correspond to the centers of unit squares in $\Lambda$ and  edges $\overline{E}$ connect any two vertices whose corresponding squares share an edge in $\Lambda$.  
An edge $e' \in \overline{E}$ that is dual to $e = (i, j) \in E$ is in $\Gamma(\sigma)$ if $J_{ij}\sigma(i)\sigma(j) = -1$ and we omit it if $J_{ij}\sigma(i)\sigma(j) = +1$.  For the Ising model where all the $J_{ij}=+1$, the contour representation $\Gamma(\sigma)$ is precisely the set of edges separating $+$ and $-$ components in $\sigma$.
Note that we can reconstruct the spin configuration $\sigma$ from the contour representation (given a single spin) if we know the values of $\{J_{ij}\}$.  The weight of a configuration $\sigma$ is determined by $\Gamma(\sigma)$, and there is a weight-preserving bijection between the configurations of any two spin glasses with the same set of frustated vertices.

\begin{figure}[t]
\begin{center}
\includegraphics[scale=.4]{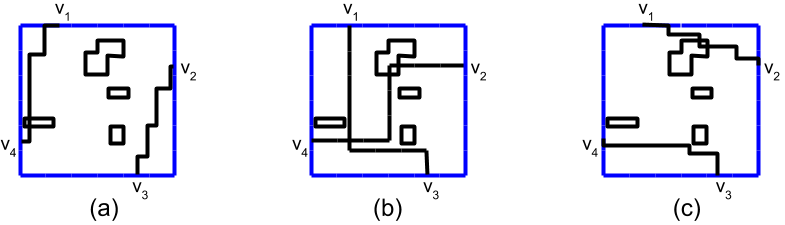}
\end{center}
\vspace{-.2in}
\caption{States with (a) positive orientation, (b) orientation 0, (c) negative orientation.}
\label{fig:orientations}
\end{figure}

For the spin glass model considered here, all vertices of $\overline{V} \setminus \{v_1,...,v_4\}$  have even degree in $\Gamma(\sigma)$ and  the frustrated vertices $\{v_1,...,v_4\}$  have odd degree.  
 It follows that $\Gamma(\sigma)$ must be the union of two paths terminating at the frustated vertices, along with even cycles.   (Note that these paths and cycles can intersect each other, and therefore are not necessarily unique.)  In all that follows, it will be convenient to shift the primal lattice $\Lambda$ by $(-1/2, -1/2)$ so that the vertices of $\overline{\Lambda}$ are integral.
Now, recall that we assume that the four frustrated vertices lie on the boundary of a $2n \times 2n$ square $S$ within $\overline{\Lambda}$ centered at $(n, n)$, and they are symmetric under rotations by $90$ degrees.  Without loss of generality, we label these so that $v_1$ lies on the top side of $S$ and is the $i^{th}$ vertex from the upper left corner for some $0 \leq i \leq n$. Setting $p = i/2n$, $v_1$ is at a distance of $2pn$ from the upper left corner, $v_2$ is on the right side of $S$ a distance of $2pn$ from the upper right corner, $v_3$ is on the bottom of $S$ a distance of $2pn$ from the lower right corner, and $v_4$ is on the left side of $S$ a distance of $2pn$ from the lower left corner.  The key to all of our arguments is how the two long paths in $\Gamma(\sigma)$ pair up these frustrated vertices.
Let $\alpha(\sigma)$ be the length of the shortest path in $\overline{\Lambda}$ from the connected component of $\Gamma(\sigma)$ containing $v_1$ to the connected component containing $v_4$ (if $v_1$ and $v_4$ are connected, $\alpha(\sigma)=0$).  Likewise, let $\beta(\sigma)$ be the length of the shortest path between the component containing $v_1$ and the component containing $v_2$.  Let $\gamma(\sigma) = \beta(\sigma)-\alpha(\sigma)$ be the {\it orientation} of the configuration $\sigma$.  We  partition the state space $\Omega$ into a disjoint union  $\Omega = \cup_{i \in \mathbb{Z}} \  \Omega_i$, where $\sigma \in \Omega_i$ if $\gamma(\sigma) = i$.  

The partition of $\Omega$ into $\cup_i \Omega_i$ allows us to define a cut in the state space in order to bound the conductance.
In particular, we let $\Omega^- = \cup_{i<0} \ \Omega_i$ and $\Omega^+ = \cup_{i>0} \ \Omega_i$, and we observe that $\Omega = \Omega^- \cup \Omega_0 \cup \Omega_+$.  We specify a subset of $\cC \subset \Omega_0$ that will be critical to defining the cut as $\cC =  \{\sigma \in \Omega_0: \alpha(\sigma) = \beta(\sigma) = 0\}$ (i.e., the configurations in which $v_1$ is connected to both $v_2$ and $v_4$). See Figure~\ref{fig:orientations}.  Finally, we define $\cC^* = \cC \cup \Omega_{-1} \cup \Omega_1$ to be the configurations where the paths connecting the frustrated vertices are within distance 1 of each other.   Following \cite{randall-top}, for configurations in $\cC$, we partition the cross into two paths, one from $v_1$ to $v_3$ and a one from $v_2$ to $v_4$; we do the same for configurations in $\Omega_{-1}$ and $\Omega_1$, although it may be necessary to add a single ``defect'' that encodes where one or both of these paths incurs a jump by one unit.
To move from a configuration in $\Omega^-$ to one in $\Omega^+$  using Glauber dynamics, we must pass through a configuration in $\cC^*$.
We will show that the probability of $\cC^*$ is exponentially small, and this will allow us to argue that the Glauber dynamics requires exponential time to converge to equilibrium.

\section{Slow Mixing for the Ising model with Mixed Boundaries}\label{sec:cross}

We start with the standard approach used to show slow mixing when the boundary conditions alternate spins on the boundary of a $(2n+1) \times (2n+1)$ lattice region $\Lambda$.  
Here $\overline{\Lambda}$ is the $2n \times 2n$ lattice region centered in $\Lambda$.
This will motivate the approach used in the general spin glass setting (when the frustrated vertices are not necessarily on the boundary of $\Lambda$) and will elucidate the difficulties in generalizing this simpler result.  

Fix $0\leq p \leq 1/2$ and let $q=1-p$.  We define $v_1=(2pn, 2n), v_2=(2n,2qn), v_3=(2qn,0)$ and $v_4=(0,2pn)$.  We consider {\it p-shifted mixed boundary conditions} in which all vertices on the boundary between $v_1$ and $v_2$ and between $v_3$ and $v_4$ are assigned $+$ and the others are assigned $-$.  
The vertices $v_1,...,v_4$ define the endpoints of a pair of paths in each configuration.  
(There may be more than one choice of paths.) 
Using the strategy outlined in Section~\ref{sec:contour}, we recall that $\cC$ consists of those configurations where there are paths from $v_1$ to both $v_2$ and $v_4$ (and therefore also to $v_3$).  Using the notion of 	``fault lines'' introduced in \cite{randall-top}, we note that this is the set of configurations that contain a {\it horizontal fault line}, i.e.,. a path from $v_2$ to $v_4$, and a {\it vertical fault line}, i.e., a path from $v_1$ to $v_3$.  When both fault lines are present (and intersect) we call their union a {\it cross}.  We define the cross so that it is a maximal component of the contour representation of the configuration.

Let $C$ be a cross in $\overline{\Lambda}$.  We have the following lemma. See Figure~\ref{fig:min-cross-a}.

\begin{lemma}
\label{lem:min-cross}
Let $S_n$ be the $2n\times2n$ axis-aligned square whose sides contain $v_1,..,v_4$. 
Then $|C| \geq 8n-8pn$. 
\end{lemma}

\begin{proof}
Remove from $C$ a simple path from $v_1$ to $v_3$. What's left is a graph of even degree at every vertex except $v_2$ and $v_4$. $v2$ and $v4$ are in the same component of this graph, because the total degree of a component is even. Therefore there is a path from $v_2$ to $v_4$. The two paths are disjoint and each contains at least $4n - 4pn$ edges.
\end{proof}

Let $L = 8n-8np$. We write the length as $|C|=L + \ell$, for some $\ell \geq 0$.  Let $\cC_C $ be the set of configurations in $\cC$ that have $C$ as their cross. 

We will write the weight of a configuration $\sigma$ as $\lambda^{-H(\sigma)}$, $\lambda = e^{2\beta} = e^{2/T}$, and note that the energy $H(\sigma)$ is the number of edges in the contour representation of $\sigma$. 

\begin{lemma}\label{lem:cross} For any  cross $C$, we have
$$\pi(\cC_C ) \leq \lambda^{-(4n-8pn+\ell)}.$$
\end{lemma}

\begin{proof}  
We define an injective map $\psi_C: \cC_C \rightarrow \Omega$ so that
$\pi(\psi_C(\sigma)) \geq \pi(\sigma) \lambda^{(L-4n+\ell)}$.
Given this  map, we find
$$
1 \ =\ \pi(\Omega)  \geq  \sum_{\sigma\in\cC_C } \pi(\psi_C(\sigma))
\ \geq \  \sum_{\sigma\in\cC_C } \pi(\sigma) \lambda^{(L-4n +\ell)} 
\ = \ \lambda^{(4n-8pn+\ell)} \pi(\cC_C ).
$$

\noindent The map $\psi_C$ is defined by removing $C$; 
then, along the upper-left boundary of $\Lambda$ between $v_1$ and $v_4$ we add each edge not in $\sigma \setminus C$ and remove each edge in $\sigma \setminus C$; then, along the lower-right boundary of $\Lambda$ between $v_3$ and $v_2$ we add each edge not in $\sigma \setminus C$ and remove each edge in $\sigma \setminus C$.
\end{proof}

\begin{theorem}\label{thm:peierls}
Let $\Lambda \subset {\mathbb{Z}}^2$ be an $(2n+1) \times (2n+1)$ lattice region and $0 \leq p \leq 1/2$ define a family of $p$-shifted mixed boundary conditions on $\Lambda$.  
Let  $\Omega$ be the set of all Ising configurations and let $\cC$ be the Ising configurations containing  a cross.  Then
$$\pi(\cC) \leq f(n) e^{-cn},$$
for some polynomial $f(n)$ and constant $c>0$, whenever $\lambda^{(1-2p)} > 3^{(2-2p)}.$
\end{theorem}

\begin{proof}  
By Lemma~\ref{lem:cross},
$$
\pi(\cC)
\ \leq \ \sum_C \lambda^{-(4n-8pn+\ell)} 
\  \leq \ \sum_{\ell \geq 0} \lambda^{-(4n-8np+\ell)} 3^{(8n-8np+ \ell)}
 \leq \ 4n^2 (3^{(2-2p)}\lambda^{-(1-2p)})^{4n},$$
which is exponentially small when $\lambda^{(1-2p)} > 3^{(2-2p)}.$
The second inequality holds because there are at most $3^{(8n-8np+ \ell)}$ ways to choose a cross of length 
$8n-8np+\ell$.
\end{proof}

\noindent Thus, when $\lambda^{(1-2p)} > 3^{(2-2p)}$ we have that the size of the cut is exponentially small, and therefore the conductance of the graph is also exponentially small.  By Theorem~\ref{CondThm}, this implies that the chain takes exponential time to mix.

\begin{corollary}
Glauber dynamics for the Ising model on $\Lambda$ with $p$-shifted mixed boundary conditions takes time at least $e^{cn}$ to mix, for some constant $c>0$, when $\lambda^{(1-2p)} > 3^{(2-2p)}.$
\end{corollary}

\noindent Notice that this gives $\lambda > 9$ when $p=0$ and $\lambda > 3^{(2^{(k-1)}+1)}$
 when $p=1/2-1/2^k$ and
 when $p=1/2$ this fails to give any useful bound.

\section{Slow Mixing for Frustrated Spin Glasses Using Free Energy}\label{sec:slow}
We will now proceed to extend the result in Section~\ref{sec:cross} by establishing slow mixing below some temperature for spin glasses with four well-separated frustrated vertices.   

In this setting we define $\Lambda$ as a $kn \times kn$ lattice region, $k \geq 2$. 
Four distinguished faces are symmetric around the center of $\Lambda$ under $90$ degree rotations. The centers of these faces are four vertices $v_1,..,v_4$ that again lie on the boundary of a $2n \times 2n$ square $S$ within $\overline{\Lambda}$ centered at $(n, n)$. 
As in Section~\ref{sec:contour} we define $\cC$ to be the set of contour configurations in which $v_1$ is connected to both $v_2$ and $v_4$, and we define the cross in such a configuration as the component containing $v_1$. 
We define $\cC_C$ to be the set of contour configurations whose cross is $C$.
The argument in Section~\ref{sec:cross} fails when $p = 1/2$, in particular when $\ell = o(n)$.
The length of the cross $C$ in that case is $4n + \ell$,  
and our injective map $\psi_C$ removes $C$ and replaces it with two paths of total length $4n$. 
The loss of energy, $H(\sigma) - H(\psi_C(\sigma)) = \ell$, is too small to show that $\sigma$ has exponentially small probability.

The remedy comes from noticing that in exactly the case $\ell = o(n)$, $C$ is nearly a minimal cross 
and there are many alternative choices of $\psi_C$. Recall that $\psi_C$ was defined by removing $C$ and adding paths along the boundary of the $2n\times 2n$ square from $v_1$ to $v_4$ and from $v_3$ to $v_2$.

Instead of restricting ourselves to paths along the boundary of the square, we will allow any monotone ``staircase'' path, with the restriction that the path does not intersect the axis-aligned square region $S_{\ell}$ of side length $2n - 4np + \ell$ centered at the center of $\Lambda$.  
We choose this restriction to ensure loss of energy as a consequence of Lemma~\ref{lem:staircase-cross} below. 
As shown in Figure~\ref{fig:min-cross-b}, when $\ell=0$, we allow any monotone path in one of the shaded rectangles of dimensions $2np\times (2n-2np)$, and as $\ell$ increases the monotone paths must stay in the shaded region that avoids $S_{\ell}$.

Our new strategy is to use {\it all} possible choices of $\psi_C$ for small $\ell$, thereby defining an exponential family of images of each configuration $\sigma \in \cC_C $. We will define a function $\Psi_C$ that maps $\sigma$ to the union of possible $\psi_C(\sigma)$ defined by different pairs of monotone paths. Figure~\ref{fig:min-cross-b} also shows the tradeoff between energy and entropy for our method. 
As $\ell$ increases and $p$ decreases, the energy loss due to the map increases. 
As the size of each shaded area decreases, the number of possible paths, bounded above by $\binom{2n}{2np}$, also decreases. 
This is what we mean by a decrease in entropy. 

Just as we needed $\psi_C$ to be injective in Section~\ref{sec:cross}, we would hope for our new map to have the property that two different configurations map to disjoint sets of configurations. 
We were not able to find such a map. 
However, we are able to define $\Psi_C$ to pass a small amount of ``side information,'' and with this definition we will get a disjointness property that serves our purpose.
The side information is in the form of tokens placed on certain edges along each of the two paths that define the configuration $\sigma$ is mapped to. Formally, for each path this information is encoded as a binary string of length $2n$: \texttt{0} for any plain edge, \texttt{1} for an edge with a token. The nice property that will make this side information small is that no two adjacent edges of a path are occupied by tokens.

Let $B(m)$ be the set of binary strings of length $m$ with no consecutive \texttt{1}'s. 
Let $B = B_C = B(2n)$.
Formally, we will define a function $\Psi_C: \cC_C  \rightarrow 2^{\Omega \times B \times B}$
that has the nice properties in the following lemma. 
To get our hands on the set of mapped configurations minus the tokens, 
we define the projection operator $\Pi: 2^{\Omega \times B \times B} \rightarrow 2^{\Omega}$, 
so that $\Pi(\{\sigma_i, b_i, b_i'\}) = \{\sigma_i\}$. 
Informally, $\Pi \circ \Psi_C$ is the map from one configuration to a set of configurations.

\begin{figure}[t]
\centering
 \begin{subfigure}{0.4 \linewidth}
    \includegraphics[scale=.25]{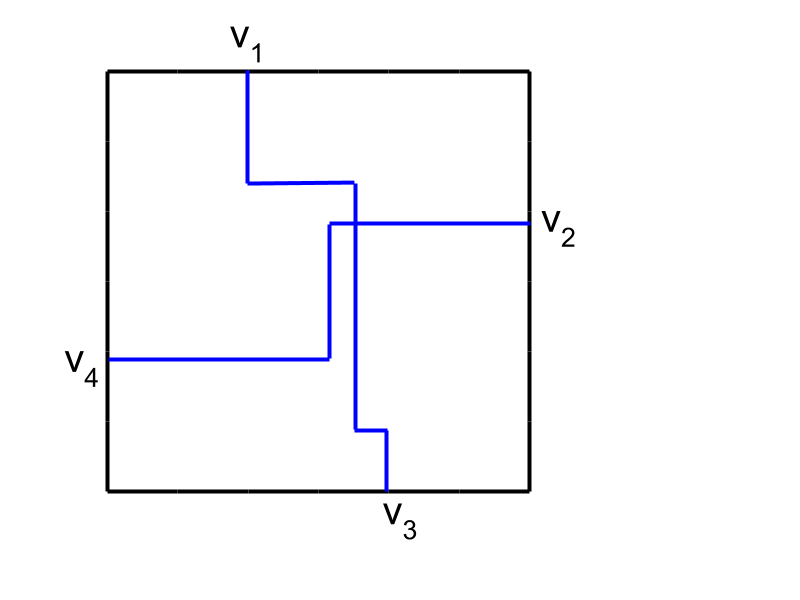}
    \subcaption{\centering}
    \label{fig:min-cross-a}
  \end{subfigure}
  \begin{subfigure}{0.4 \linewidth}
    \includegraphics[scale=.25]{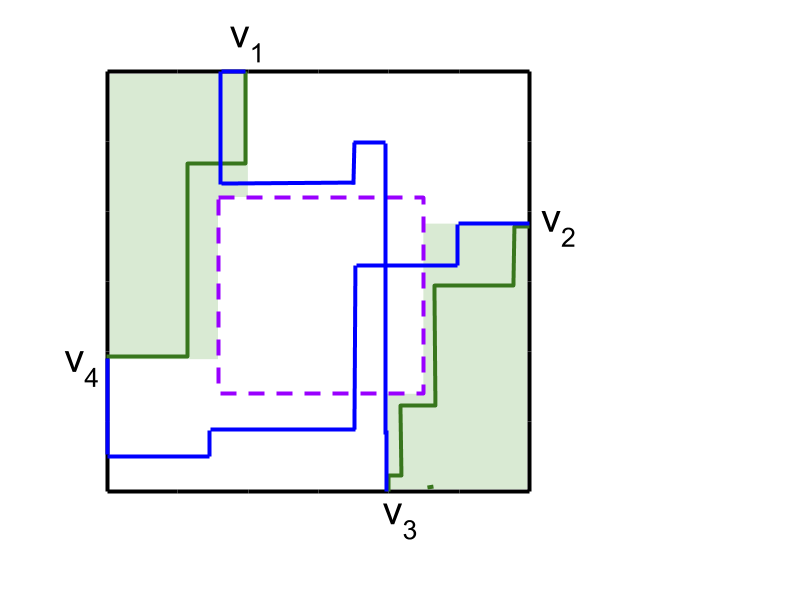}
    \subcaption{\centering}
    \label{fig:min-cross-b}
  \end{subfigure}
\caption{(a) A minimal cross, of length $L$. (b) A cross of length $L+\ell$ in blue, with $S_{\ell}$ in purple and two possible monotone paths in green. Any monotone path in either shaded region is possible.}
\label{fig:min-cross}
\end{figure}

In the following lemmas, fix $0\leq p \leq 1/2$ and let $L = 8n-8pn$. 
For $b<0$ let $\binom{a}{b} = 1$.

First we give a lower bound on the number of pairs of monotone paths, which we call {\em staircases}.
\begin{definition}
An {\em upper staircase} in $\overline{\Lambda}$ is a path of $2n$ west and south edges starting at $v_1$ and ending at $v_4$. A {\em lower staircase} is a path from $v_3$ to $v_2$ which, when the configuration is rotated $180\degree$, becomes an upper staircase.
\end{definition}
Note that the edges on a staircase need not be edges of a particular configuration. 
For the rest of the section we define $h(x) = -x\log_2(x) - (1-x)\log_2(1-x)$.

\begin{lemma}\label{lem_num_staircases}
Let $\ell \leq 4np^2$. 
The number of pairs of staircases, one upper, one lower, that do not intersect $S_{\ell}$ is at least 
$2^{4n h(p)^4}\mathrm{poly}(n)$.
\end{lemma}
\begin{proof}
We will show that the number of upper staircases is at least $2^{2n h(p)^2}\mathrm{poly}(n)$.
The bound on lower staircases is the same.
Recall that $v_1$ and $v_4$ are on the boundary of a $2n\times 2n$ square centered at $(n, n)$, and the corner of the square between $v_1$ and $v_4$ is the vertex $(0, 2n)$. 
Given the bound on $\ell$, $S_{\ell}$ does not contain the vertex $x = (2np(1-p), 2n - 2np(1-p))$.
The number of upper staircases is at least the number of upper staircases that contain $x$. 
A monotone path from $v_1$ to $x$ consists of $2n - 2np$ edges of which $2np(1-p)$ are south edges, 
and a monotone path from $x$ to $v_4$ consists of $2np$ edges of which $2np^2$ are south edges.
The number of such staircases is
\begin{equation}
\binom{2n(1-p)}{2np(1-p)}\binom{2np}{2np^2} = 2^{4n h(p)^4}\mathrm{poly}(n),
\end{equation}
using Stirling's formula.
\end{proof}

\begin{lemma}\label{lem:mapping-bound}
Let $C$ be a cross of length $|C| = L + \ell$ with $\ell \leq 4np^2$.
There exists a function $\Psi_C: \cC_C  \rightarrow 2^{\Omega \times B \times B}$ such that 
$\forall \sigma, \sigma' \in \cC_C $, $\sigma'' \in \Pi\circ\Psi_C(\sigma)$,
$$
 \Psi_C(\sigma) \cap \Psi_C(\sigma') \ = \ \emptyset, $$
$$ |\Psi_C(\sigma)| \ \geq \  2^{4n h(p)^4}\mathrm{poly}(n), $$ 
$$ {\rm and} \  \ H(\sigma'') \ \leq \ H(\sigma) - (4n-8np+\ell/2),$$
\end{lemma}

We postpone constructing the function $\Psi_C$ 
(and proving Lemma~\ref{lem:mapping-bound}) until the next subsection. 
Theorem~\ref{thm:all-p} is an analogue of Theorem~\ref{thm:peierls} that gives an exponential
bound for all $p,~0\leq p\leq 1/2$. 
As a corollary of Theorem~\ref{thm:all-p}, we will prove our main result, Theorem~\ref{thm:main},
 asserting slow mixing for spin glasses with frustration.

From now on we drop the $\mathrm{poly}(n)$ from our calculations.
We first bound the probability of the set of configurations containing a given cross $C$.
For use in the next two lemmas we define
\begin{equation*}
f(\ell, p) =
\begin{cases}
  4n (h(p)^4 - \log_2\phi), ~ \ell \leq 4np^2 \\
  \ell/2 \log_2\lambda, ~ \ell > 4np^2 ,
\end{cases}
\end{equation*}

\begin{lemma}\label{lem:piOmegaC}
For any cross $C$ of length $|C| = L + \ell$   we have
\begin{equation*}
\pi(\cC_C ) \leq  \lambda^{-(4n-8np+\ell/2)} 2^{-f(\ell, p)},
\end{equation*}
where 
$\phi = {(1+\sqrt{5})}/{2}$.
\end{lemma}

\begin{proof}
For the case of $\ell > 4np^2$, the proof of Lemma~\ref{lem:cross} carries over word for word. 
We concentrate on the case of $\ell \leq 4np^2$.

It is well known that $B = B(2n)$ is the $2n^{th}$ Fibonacci number, which is less than~$\phi^{2n}$.
Each $\sigma'' \in \Pi\circ\Psi_C(\sigma)$ appears in at most $B^2 \leq \phi^{4n}$ 
elements of $\Psi_C(\sigma)$. 
The upper bound on $H(\sigma'')$ in Lemma~\ref{lem:mapping-bound}, gives $\pi(\sigma'') \geq \pi(\sigma)\lambda^{-(4n-8np+\ell/2)}$.
The lower bound on $|\Psi_C(\sigma)|$ and the disjointness property imply
\begin{align}
1 \geq \pi(\Pi\circ\Psi_C(\cC_C )) & \geq \sum_{\sigma\in\cC_C } \pi(\sigma)  \lambda^{(4n-8np+\ell/2)} 
      \phi^{-4n}  2^{4n h(p)^4}.
\end{align}
The inequality follows by replacing $\sum \pi(\sigma)$ with $\pi(\cC_C )$.
\end{proof}

Theorem~\ref{thm:all-p} depends on the following technical lemma 
regarding the set $\cC_{\ell}$ of configurations 
containing crosses of fixed length $L + \ell$:
$\cC_{\ell} = \union \{\cC_C : |C| = L + \ell \} $. 
The idea of the lemma is to show that $\pi(\cC_{\ell})$ is exponentially small, 
where the constant in the exponent is independent of $\ell$. 
This also means that {the free energy} $\ln\pi(\cC_{\ell})/n$ is less than
some negative constant. 
Since there are polynomially many values of $\ell$, 
it will follow that the whole set $\cC$ is exponentially small.

\begin{lemma}\label{lem:piOmegaL1}
Let $\cC_{\ell}$ be the set of contour configurations containing a cross of length $L + \ell$. 
Then for large enough $\lambda$ 
\begin{equation}\label{eqn:piOmega0L}
- \log_2(\pi(\cC_{\ell}))/n = \Omega(1).
\end{equation}
\end{lemma}

\begin{proof}
Let $s = 1/2 - p$, $r = \ell/n$. Abusing notation, let $f(r, s) = f(\ell, p)/n$.

As $s, r \rightarrow 0$, $f(r, s) \rightarrow 4 - 4\log_2\phi \approx 1.2$.

Each $C \in \cC_{\ell}$ contains a minimal cross consisting of a vertical path connecting $v_1$ to $v_3$ and a horizontal path connecting $v_2$ to $v_4$. 
The vertical path contains $2n$ vertical edges and $2n-4pn$ horizontal edges. 
There are at most $\binom{4n-4pn}{2n-4pn} = \binom{2n+4sn}{4sn}$ choices of  vertical path.  
Likewise there are at most $\binom{2n+4sn}{4sn}$ choices of minimal horizontal path.  
Then there are $\binom{8n-8np+\ell}{\ell} = \binom{4n+8sn+rn}{rn}$ ways to choose the locations of the $\ell$ extra edges, and $3$ possible directions for each extra edge. In sum,

\begin{align*}
\pi(\cC_{\ell}) & \leq | \{C : |C| = L + \ell \}| \max_{|C| = L + \ell} \pi(\cC_C ) \\
      & \leq \binom{4n+8sn+rn}{rn} \binom{2n+4sn}{4sn}^2 3^{rn} 
      \max_{|C| = L + \ell} \pi(\cC_C ) \\
      & \leq \exp_2\left[n(4+8s+r)h\left(\frac{r}{4+8s+r}\right) + n(4+8s)h\left(\frac{2s}{1 + 2s}\right)\right] 3^{rn} \\
      &  \qquad \cdot \max_{|C| = L + \ell} \pi(\cC_C ) \\
      & \leq \exp_2[n(~ (4+8s+r)h(r/(4+8s+r)) + (4+8s)h(2s/(1 + 2s)) + r\log_2 3  \\
      & \qquad\qquad  - \log_2\lambda(8s + r/2) - f(r, s)~)] ,
\end{align*}
where we have applied first Stirling's formula and then Lemma~\ref{lem:piOmegaC} in the last two inequalities.

Taking $\log_2$ of both sides, dividing by $n$, and rearranging terms, we have
\begin{align*}
- \log_2(\pi(\cC_{\ell}))/n & \geq  
                -[~(4+8s+r)h(\frac{r}{4+8s+r}) + (4+8s)h(\frac{2s}{1+2s}) + r \log_2 3 ~] \cr
				     & \qquad + [~ (8s+ r/2) \log_2 \lambda + f(r, s) ~] .
\end{align*}

For small $s$ and $r$ the first term is $o(1)$ and the second term is $\Omega(1)$. Otherwise, for large $\lambda$ the second term dominates.
\end{proof}

We now state the key theorem bounding  the probability of the set $\cC$ of configurations containing crosses.

\begin{theorem}\label{thm:all-p}
Let  $\Omega$ be the set of all spin glass configurations in a $kn\times kn$ square lattice $\Lambda$  centered at $(n, n)$, $k \geq 2$. Suppose that four distinguished vertices $v_1,..,v_4$ lie on the boundary of an axis-aligned $2n \times 2n$ square $S$ centered in $\overline{\Lambda}$, and these four vertices form the corners of a (not necessarily axis-aligned) square (i.e., they are shifted by $2p$ around the boundary of $S$).
Let $\cC$ be the set of configurations in which $v_1$ is connected to both $v_2$ and $v_4$. 
Then for $\lambda$ large enough we have
\begin{equation}\label{eqn:piOmega0}
\pi(\cC) \leq 2^{-\Omega(n)}.
\end{equation}
\end{theorem}
\begin{proof}
Since $\ell$ takes $O(n^2)$ values,
$\pi(\cC) \leq O(n^2) \max_{\ell} \pi(\cC_{\ell}) \leq 2^{-\Omega(n)}.$
\end{proof}

\begin{proof}[Proof of Theorem~\ref{thm:main}]
Set  $T = 2/\log_2\lambda$ for large enough $\lambda$. Let $t<T$. 
The state space $\Omega$ contains the two disjoint subsets $\Omega_-$ and $\Omega_+$, separated by a cut set $\cC^*$ consisting of all configurations within two steps of $\cC$. 
We have $\pi(\cC^*) < \pi(\cC) \mathrm{poly}(n)$ and by symmetry $\pi(\Omega_-) = \pi(\Omega_+)$. 
The conductance $\Phi$ satisfies 
\begin{equation}
\Phi \ \leq \ 
      \frac{\sum_{\sigma \in \Omega_-, \sigma' \in \Omega_{\cC}} \pi(\sigma) \Pr(\sigma, \sigma')} 
     {\pi(\Omega_-)}
    \  \leq \ 4\cdot \pi(\cC^*)
    \ \leq \ 2^{-\Omega(n)},~~ {\mathrm{for~large~enough}}~ n. 
\end{equation}
Therefore the Markov chain mixes slowly. 
\end{proof}

\subsection{Construction of the Map}\label{map_intro}

In this section we will construct the map $\Psi_C$
using pairs of paths as shown in Figure~\ref{fig:min-cross-b}. 

Given upper and lower staircases, we will map $\sigma \in \cC_C $ to some $\sigma' \in \Omega$,  marking certain edges with tokens. 
We will show that one can reconstruct $\sigma$ from $C$, $\sigma'$, and the marked edges, no two marked edges are adjacent, and  
$H(\sigma') \leq H(\sigma) - |C| + 4n$, implying Lemma~\ref{lem:mapping-bound}.

We briefly outline our approach here. The map is motivated by the map $\psi_C$ in the proof of Lemma~\ref{lem:cross}, but the map must encode the locations of the staircase edges in $\sigma'$ without increasing $H(\sigma')$. 
The first idea for the map is to remove $C$ from $\sigma$ then add the edges of the staircases. 
The result is nearly a configuration in $\Omega$, except that some edges of $\sigma$ are doubled  by edges of the staircases. 
The first step in removing double edges is to shift components of $\sigma$ lying between one of the staircases and $C$, away from the staircase toward the removed edges of $C$. 
See Figure~\ref{fig:staircase-cross}. 
In case a segment of the staircase bisects a component of $\sigma$ we must shift that segment to preserve the even degrees of vertices.
The second step is to remove both copies of any double edges.
If we do this with every double edge it may happen that $\sigma$ cannot be reconstructed;
so, at the corners of the staircase we divert the staircase to avoid $\sigma$ and place a token to mark the diverted edge. 
See Figure~\ref{fig:staircase}.

Let $S_U$ be an upper staircase and $S_L$ be a lower staircase. 
The simple regions in the interior of $C \cup S_U \cup S_L$ may be two-colored gray and white, with the exterior colored gray.  

Regions separated by a path in $C \cap S_U$ or $C \cap S_L$ will have the same color. (We can think of doubled paths as regions of zero area.)
This situation is shown in Figure~\ref{fig:staircase-cross}. 

\begin{figure}[t]
\centering
\includegraphics[scale=0.4]{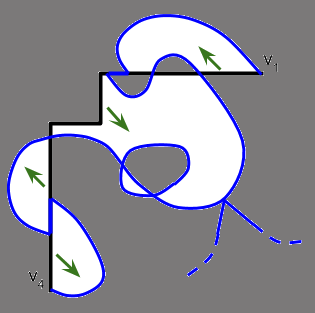} 
\caption{A staircase in black, with the part of a cross containing a path from $v_1$ to $v_4$ in blue. The green arrows show the direction edges of $\sigma$ are shifted in each white region.}
\label{fig:staircase-cross}
\end{figure}

Recall that $S_{\ell}$ consists of the interior and boundary edges of the $(2n-4np+\ell)\times(2n-4np+\ell)$ axis-aligned square centered at the center of $\Lambda$, as shown in Figure~\ref{fig:min-cross-b}.
We have the following lemma.

\begin{lemma}
\label{lem:staircase-cross}
If  $(S_U \cup S_L) \cap S_{\ell} = \emptyset$, then no simple region in the interior of $C \cup S_U \cup S_L$ has a boundary containing edges of both $S_U$ and $S_L$. 
\end{lemma}

\begin{proof}
Suppose there is such a simple region. We will show in this case that $C$ contains at least $8n-8np+2\ell = L + 2\ell$ edges, a contradiction. $C$ contains two disjoint paths $r_1$ and $r_2$ that are each part of the boundary of $R$ and connect $S_U$ to $S_L$, as well as two disjoint paths $p_1$ from $v_2$ to $v_3$, and $p_2$ from $v_1$ to $v_4$. $r_1$ and $r_2$ each have length at least one side length of $S_{\ell}$, which is $2n - 4np + \ell$. Since $r_1, r_2 \in C$ and $R$ is a simple region, the interior vertices of $r_1$ and $r_2$ have degree two. Therefore, each $r_i$ is contained in one of the $p_j$ or disjoint from both. In the disjoint case, $p_1$ and $p_2$ each contain at least $2n$ edges, and $C$ contains at least $8n-8np+2\ell$ edges. 
In the other case, say, $p_1$ contains $r_1$. In the best case $p_1$ also contains $r_2$, and traversing $p_1$ from $v_2$ to $v_3$ there are no north edges. Then there are $2n-2np$ south edges and $2(2n-4np+\ell) + 2np = 4n-6np+2\ell$ east-west edges, a total of $6n-8np+2\ell$ edges. Adding the $2n$ edges from $p_2$, $C$ contains at least $8n-8np+2\ell$ edges.
\end{proof}

Let $R$ be a white simple region in the interior of $C \cup S_U \cup S_L$ whose boundary contains edges of $S_U$. 
By Lemma~\ref{lem:staircase-cross} the boundary of $R$ does not also contain edges of $S_L$.
Considering the edges of $S_U$ to be directed from $v_1$ to $v_4$,
$R$ is to the right or left of $S_U$.

We may assign a $+$ or $-$ to each site in $\Lambda$ so that the sites in $R$ adjacent to $C$ are $+$ and the edges of $\sigma$ are exactly those edges between two neighboring sites of opposite sign.
We define a {\em patch} to be a connected set of $-$ sites. 
The outer boundary of a patch is the unique cycle of edges in the
configuration that, when traversed counterclockwise, has sites inside to the
left of each edge and sites outside to the right.

We define an {\em interior edge} of $S_U$ to be one that bounds two ``$-$'' sites and a {\em boundary edge} of $S_U$ to be one that is on the boundary of a patch.

\begin{figure}[t]
\begin{center}
\includegraphics[scale=.25]{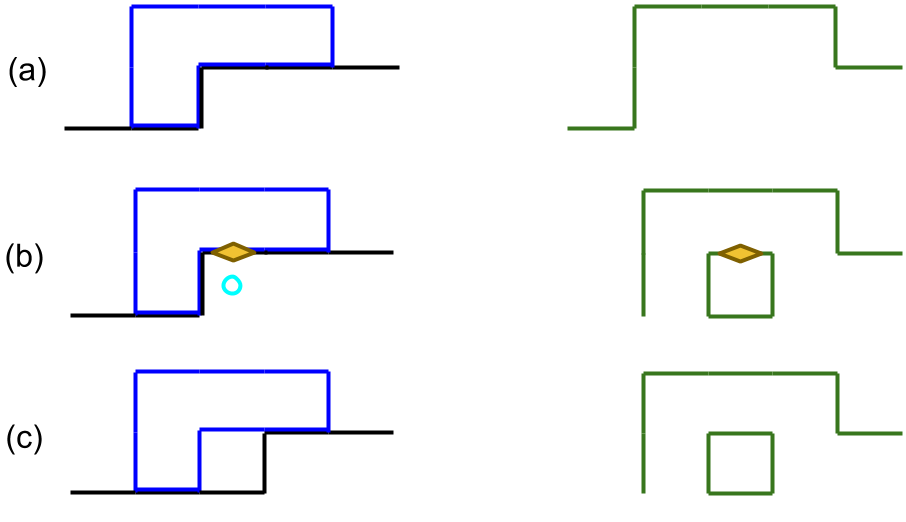} 
\end{center}
\caption{(a) A staircase and patch that share four edges (left), and an encoding that loses information (right). 
(b) A staircase and patch that share four edges (left), and an encoding with a token that preserves information (right).
(c) A staircase and patch that share two edges (left), and an encoding that preserves information (right).
}
\label{fig:staircase}
\end{figure}

We consider a white region $R$ to the left of $S_U$ and describe the map inside $R$ and along its boundary. 
The map acts symmetrically on regions to the right of $S_U$ and regions bounded by $S_L$.
Edges of $C$ are removed from the boundary of $R$.
Edges of $\sigma$ contained in $R \cup S_U$ are shifted away from $S_U$ as shown in Figure~\ref{fig:staircase-cross}: since $R$ is to the left of $S_U$ the edges are shifted one step south and one step east.

Next the edges of $S_U$ on the boundary of $R$ are added. 
Any interior added edges are shifted in the same direction as the edges of $\sigma$; this restores all the lattice points in $R \cup S_U$ to even degree as shown in Figure~\ref{fig:patch_edge_labels}.
 
Next each west internal edge of $S_U$ that precedes a south edge is marked with a token.
The sign is flipped at each site to the southeast of the corner formed by a south internal edge of $S_U$ followed by an unmarked west edge.

The drawing on the left of Figure~\ref{fig:patch_edge_labels}(a) 
shows the staircase in black and the patch
in blue before the shifting step. 
Shifting the edges of $\sigma$ and internal edges of $S_U$ and placing the token produce the drawing in the center.
Flipping the site and removing the internal edges produces the drawing on the right.

\begin{figure}[t]
\centering
\includegraphics[scale=.4]{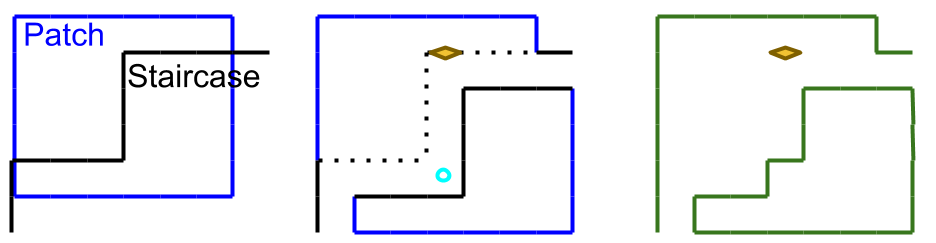}
\caption{Left: An interior section of $S_U$. Right: the map}
\label{fig:patch_edge_labels}
\end{figure}

The addition of edges of $S_U$ may have created double edges.
If the map were simply to delete double edges then the map would fail to encode the staircase edges, as shown in Figure~\ref{fig:staircase}(a). 
To avoid this problem the map operates as follows.
Consider each west double edge that is followed by some number of south double edges. 
The west edge of $S_U$ is marked with a token, shown in gold in Figure~\ref{fig:staircase}(b). 
Next the sign is flipped at the site to the east of each south double edge, marked with a blue circle in Figure~\ref{fig:staircase}(b). 
Next double edges are removed, producing the drawing on the right. 
We note that the mapping contains no more energy than the original and no two consecutive edges of $S_U$ are marked with tokens.
Figure~\ref{fig:staircase}(c) shows that the token is necessary. 
The drawing on the left maps to the figure on the right by the removal of the two double edges.

Restricted to $R$ the mapping contains no more energy than the original. 
When a site is flipped, if the south and west edges of $S_U$ are both interior edges, then two edges are removed and two are added, creating no excess.
The only worrisome case occurs when a south interior edge of $S_U$ is followed by a west non-interior edge. 
In this case when the site to the southeast is flipped, one edge is removed and three are added, creating two excess edges.
However the west edge of $S_U$ creates a double edge with the above-patch. 
The double edge is removed, canceling the two excess edges. 
See Figure~\ref{fig:mapping_steps}.

An edge of $C \cap S_U$ may be on the boundary of two white regions and may be added twice by the map. The same is true for $C \cap S_L$. This situation appears toward the left side of Figure~\ref{fig:staircase-cross}. The lemma below will show that no more than $\ell/2$ edges are added twice. The mapping removes $L + \ell = 8n-8np + \ell$ edges and adds at most $4n + \ell/2$ edges, for a net loss of at least $4n-8np+\ell/2$ edges. This will complete the construction, proving Lemma~\ref{lem:mapping-bound}.

Before stating the lemma we make some observations. $C$ contains a cross $C'$ that consists of two edge-disjoint paths from $v_1$ to $v_2$ and from $v_3$ to $v_4$ that intersect in at least one vertex. Recall that $\overline{\Lambda}$ is centered at $(n, n)$. It contains a {\em central region} of vertices
$\{(x, y): x \in [2np, 2n-2np] \mathrm{~or~} y \in [2np, 2n-2np]\}$. $C'$ is shown in blue and the central region is shown in green in Figure~\ref{fig:marked-edges-a}. Note that a horizontal line through the central region crosses an even number of vertical edges in $C'$, and a horizontal line not contained in the central region crosses an odd number of vertical edges in $C'$. A similar observation holds for vertical lines and horizontal edges of $C'$. 

We may mark $L/2$ horizontal edges and $L/2$ vertical edges of $C'$, such that a horizontal line through the central region crosses exactly two marked vertical edges, and a horizontal line not contained in the central region crosses exactly one marked vertical edge; and similarly for vertical lines and marked horizontal edges. Figure~\ref{fig:marked-edges-b} shows marked edges in pink. (The marking of edges is not unique.) Note that $C$ contains $\ell$ unmarked edges.
 
\begin{figure}[t]
\centering
 \begin{subfigure}{0.4 \linewidth}
    \includegraphics[scale=.25]{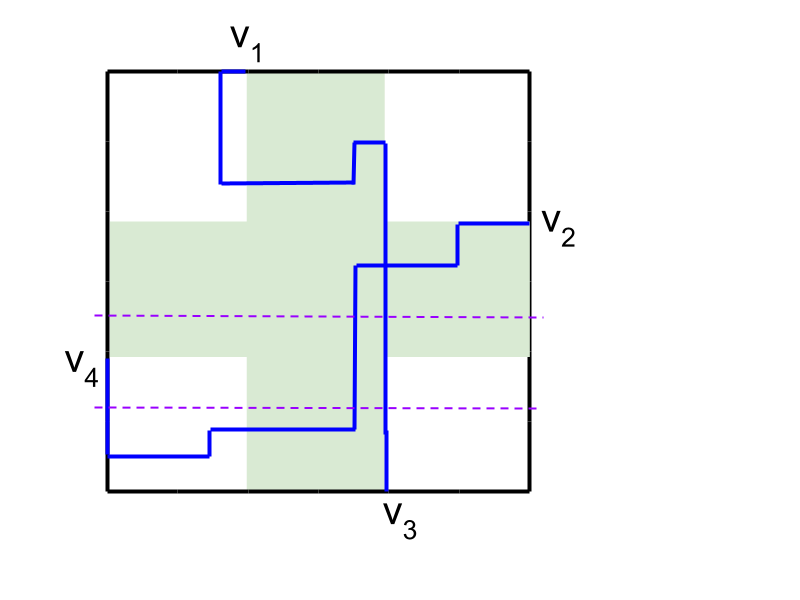}
    \subcaption{\centering}
    \label{fig:marked-edges-a}
  \end{subfigure}
  \begin{subfigure}{0.4 \linewidth}
    \includegraphics[scale=.25]{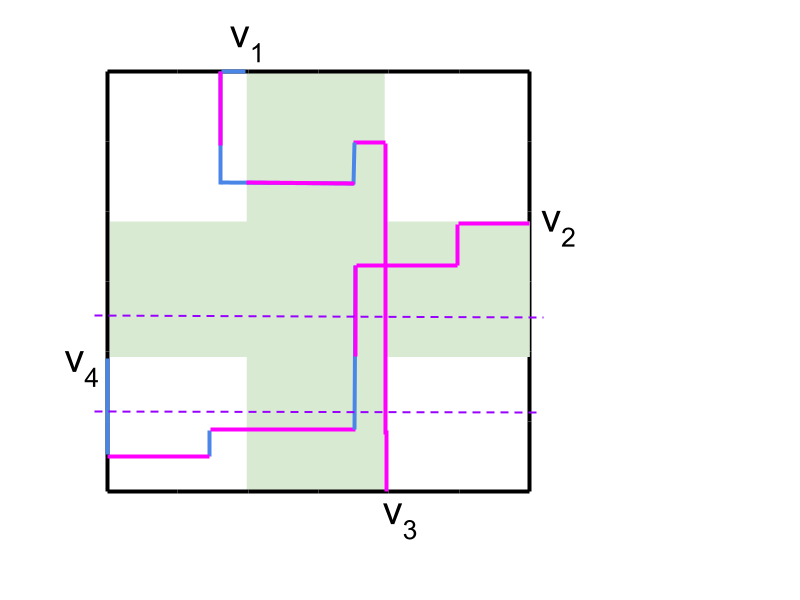}
    \subcaption{\centering}
    \label{fig:marked-edges-b}
  \end{subfigure}
\caption{Left: Two paths and central region. Right: the marked edges in pink. The upper dotted line crosses the paths twice; the lower dotted line three times.}
\label{fig:marked_edges}
\end{figure}

We can now state and prove the following lemma.

\begin{lemma}
\label{lem:white-white-edges}
Let $E$ be the set of edges in $(C \cap S_U) \cup (C \cap S_L)$ on the boundary of two white regions. 
Then $|E| \leq \ell/2$.
\end{lemma}

\begin{proof}
We will map each edge in $E$ uniquely to two unmarked edges in $C$. Referring to Figure~\ref{fig:staircase-cross}, consider an edge $e \in E$. We may assume $e$ is vertical, and it borders two white regions. A horizontal line through $e$ passes through two other edges of $C$ on the boundaries of those regions. If one of these three edges is marked, then two are unmarked. If two are marked then the horizontal line must be in the central region, and it must cross a fourth edge. Again, there are two unmarked edges. 
If the horizontal line passes through another edge in $E$, from the other staircase, it must be in the central region. It must pass through two more boundary edges for a total of six edges. At least four of these edges are unmarked. In each case, an edge in $E$ can be mapped uniquely to two unmarked edges. This proves the lemma and completes the proof of Lemma~\ref{lem:mapping-bound}.
\end{proof}

\begin{figure}[t]
\centering       
\includegraphics[scale=.3]{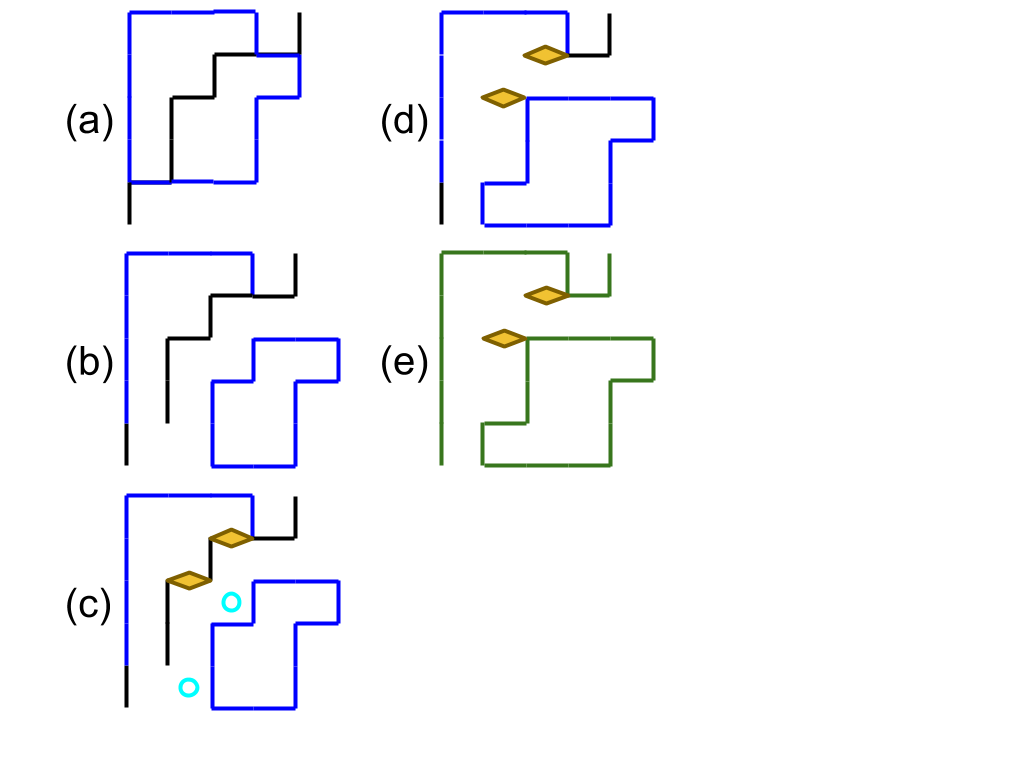}
\caption{The map: blue edges are the patch boundary, black edges are the staircase, and green edges are the final mapping.}
\label{fig:mapping_steps}
\end{figure}

We now give the algorithm for constructing the configuration $\sigma'$, one white region at a time, from the configuration $\sigma$ and two staircases $S_U$ and $S_L$. Following this we give the algorithm for reconstructing the configuration $\sigma$ from $\sigma'$ and the tokens.
Figure~\ref{fig:reconstruction_steps} shows how to reconstruct the patch and interior section of $S_U$ from Figure~\ref{fig:mapping_steps}.

\subsection*{Map Steps:}\label{appendix_mapping_steps}

Given $\sigma\in\cC_C $ pick an upper staircase $S_U$ and a lower staircase $S_L$.
Let $S$ be a subpath of $S_U$ that bounds a white region to the left.
\begin{enumerate}
\item
  Add $S$. If this doubles an edge, label one copy {\em on the staircase}
  and the other {\em above} ({\em below}) the staircase if it is on
  the boundary of an A-patch (B-patch).
\item
    Double each interior edge of $S$. Label one copy {\em on the staircase},
    and the other copy {\em below the staircase}. Shift every edge below the staircase 
    one step south and east and remove any double edges. 
    (Figure~\ref{fig:mapping_steps} step (b).)
\item 
    Place a token on each west interior edge of $S$ that is 
    followed by a south interior edge. 
    (Gold diamonds in Figure~\ref{fig:mapping_steps} step (c)).
\item 
    When a south interior edge of $S$ is followed by a west edge, 
    mark the site southeast of the corner formed by the two edges. 
    The west edge may be an interior edge or an exterior edge that was removed in Step 2.
    (Blue circles in Figure~\ref{fig:mapping_steps} step (c)).
\item
    Flip each marked site and erase the interior edges of $S$.
    (Figure~\ref{fig:mapping_steps} step (d)).
\end{enumerate}
For $S_L$, rotate the configuration $180\degree$, repeat steps 1-5, and rotate back.

\begin{figure}[t]
\centering
\includegraphics[scale=.3]{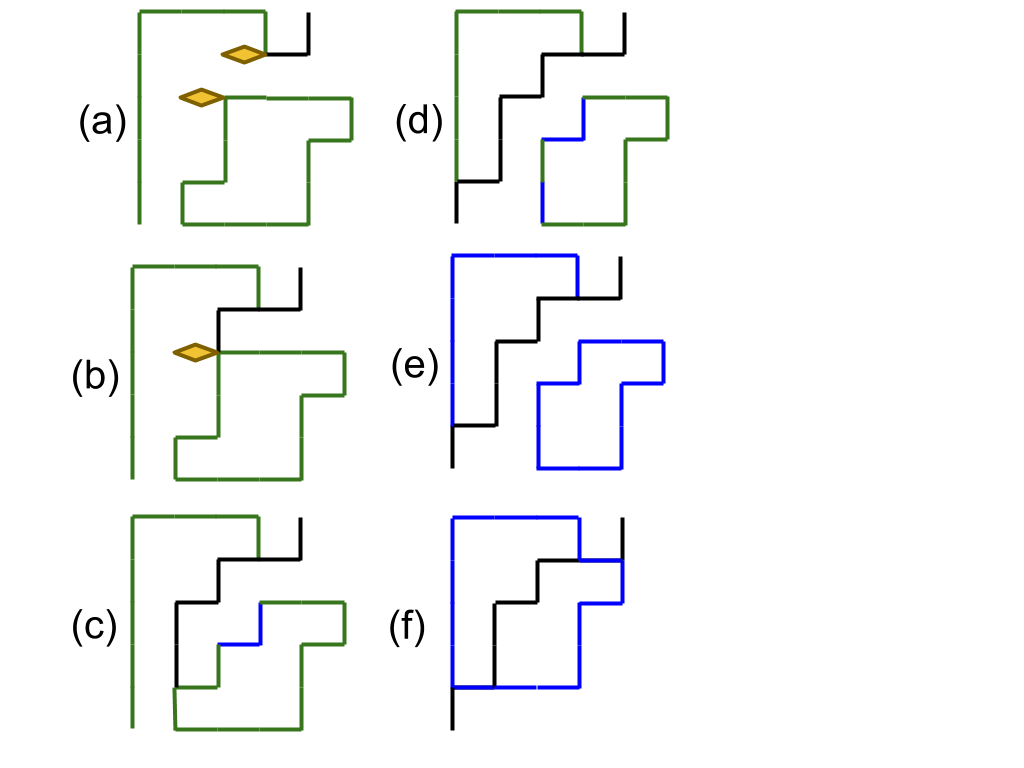}
\caption{Reconstruction steps: blue edges are the patch boundary, green edges are the mapping, and black edges are the staircase.}
\label{fig:reconstruction_steps}
\end{figure}

\subsection*{Reconstruction steps:}\label{appendix_reconstruction_steps}

Given $\sigma'$, the following steps reconstruct $\sigma$.
For a subpath $S$ of $S_U$ that bounds a white region to the left,
these steps inductively reconstruct the edges of $S$ from start to end. 
\begin{enumerate}
\item If the current edge has a token, it is a west edge and the next edge is south. 
(Figure~\ref{fig:reconstruction_steps} steps (b), (c).)
\item If the current edge is a south edge incident to two edges of $\sigma'$ that form a northwest corner, 
  \begin{enumerate}
  \item flip the site bounded by the two edges of $\sigma'$;
  \item the next edge is west. 
  \end{enumerate}
(Figure~\ref{fig:reconstruction_steps} steps (c), (d).)
\item Otherwise the next edge is in the same direction as the current edge.
(Figure~\ref{fig:reconstruction_steps} step (c).)
\item Shift every edge in the white region to the left one step north and west.
(Figure~\ref{fig:reconstruction_steps} step (e).)
\end{enumerate}
For subpaths of the upper staircase that bound a white region to the right, reflect $\sigma$ across the line $y=x$, apply steps 1-3, and reflect back. For the lower staircase, rotate the configuration 180 degrees, repeat the process, and rotate back.

\end{document}